\documentclass[11pt]{article}
    
\usepackage[utf8]{inputenc}

\usepackage{mathrsfs}
\usepackage{bbm}
\usepackage{palatino}
\usepackage[margin=3cm]{geometry}

\usepackage{amsfonts}
\usepackage{amsmath}
\usepackage{amssymb}
\usepackage{amsthm}  
\usepackage{ifthen}
\usepackage{mathabx}
\usepackage{stmaryrd}
\usepackage[colorlinks=true,linkcolor=black,citecolor=black,plainpages=false,pdfpagelabels]{hyperref}
\hypersetup{pdftitle={Generalized Entropies}}
\usepackage{tikz}
\usepackage{verbatim}
\usepackage{pgfplots}
\usetikzlibrary{shapes.geometric,plotmarks,backgrounds,fit}
\usepackage{todonotes}
\usepackage{enumerate} % Ph.: for \begin{enumerate}[(i)] ... \end{enumerate}

\usepackage{braket}

\newtheorem{theorem}{Theorem}[section]
\newtheorem{prop}{Proposition}[section]
\newtheorem{lemma}{Lemma}[section]
\newtheorem{corollary}{Corollary}[section]

\begin{document}

\title{Generalized Entropies}
%\footnote{\copyright{} by the authors}}

%\author{F. Dupuis, L. Kr\"amer, P. Faist, J.\ M.\ Renes, and R. Renner}
%\affiliation{Institute for Theoretical Physics,\\
%ETH Zurich, Switzerland\\
%\{dupuis,lkraemer,pfaist,renes,renner\}@phys.ethz.ch\\
%www.qit.ethz.ch}

\author{F. Dupuis$^{1,2*}$
\quad\quad
L. Krämer$^{1**}$
\quad\quad
P. Faist$^{1**}$
\quad\quad
J. M. Renes$^{1**}$
\quad\quad
R. Renner$^{1**}$\\[4mm]
{\it\small $^{1}$Institute for Theoretical Physics, ETH Zurich, Switzerland}\\[-1mm]
{\it\small $^{2}$Department of Computer Science, Aarhus University, Denmark}\\[-1mm]
{\tt\small $^{*}$dupuis@cs.au.dk}\\[-1mm]
{\tt\small $^{**}$\{lkraemer,pfaist,renes,renner\}@phys.ethz.ch}\\[-1mm]
{\tt\small www.qit.ethz.ch}
}
%\date{13 November 2012}
\date{}

%\keywords{Entropy, Quantum Information Theory, Thermodynamics,
%  Hypothesis Testing.}

%\bodymatter
\maketitle

\begin{abstract}
  We study an entropy measure for quantum systems that generalizes the
  \emph{von Neumann entropy} as well as its classical counterpart, the
  \emph{Gibbs} or \emph{Shannon entropy}. The entropy measure is based on hypothesis testing and has an elegant formulation as a semidefinite program, a type of convex optimization. After establishing a few
  basic properties, we prove upper and lower bounds in terms of the \emph{smooth entropies}, a family of entropy
  measures that is used to characterize a wide range of operational
  quantities. From the formulation as a semidefinite program, we also prove a result on decomposition of hypothesis tests, which leads to a chain rule for the entropy. 
\end{abstract}

% End needed for arxiv

\section{Introduction}\label{aba:sec1}

\emph{Entropy}, originally introduced in thermodynamics, is nowadays
recognized as a rather universal concept with a variety of uses,
ranging from physics and chemistry to information theory and the
theory of computation. Besides the role it plays for foundational
questions, it is also relevant for applications. For example, entropy
is used to study the efficiency of steam engines, but it also occurs
in formulae for the data transmission capacity of optical fibres.

While entropy can be defined in various ways, a very common form
employed for the study of classical systems is the \emph{Gibbs
  entropy} or, in the context of information theory, the \emph{Shannon
  entropy}~\cite{shannon_mathematical_1948}. It is defined for any
probability distribution $P$ as
\begin{align*}
  H(P) = - \sum_x P(x) \log P(x)
\end{align*}
(up to an unimportant proportionality factor). This definition has
been generalized to the von Neumann entropy~\cite{neumann_mathematical_1996}, which is
defined for density
operators,
\begin{align*}
  H(\rho) = -\operatorname{tr}(\rho \log \rho).
\end{align*}
While these entropy measures have a wide range of
applications, it has recently become apparent that they are not suitable for correctly
characterizing operationally relevant quantities in general scenarios
(as explained below). This has led to the development of
extensions~\cite{ogawa_strong_2000}, among them the \emph{information spectrum
  approach}~\cite{Nagaoka2007,han_information-spectrum_2002,Bowen2006}  and
\emph{smooth entropies}~\cite{Renner2004,Renner2005} (where the former can be obtained as an asymptotic limit of the latter~\cite{DatRen09}).

The aim of this work is to study an alternative measure of entropy
that generalizes von Neumann entropy. The generalized entropy is
closely related to smooth entropies, which, in turn, are connected to
a variety of operational quantities.

\subsection{Axiomatic and operational approach to
  entropy} \label{sec_opapproach}

The variety of areas and applications where entropies are used is
impressive, and one may wonder what it is that makes entropy such a
versatile concept.

One could attempt to answer the question from an \emph{axiomatic}
viewpoint. Here, the idea is to consider (small) sets of axioms that
characterize the nature of entropy. There is a vast amount of
literature devoted to the specification of such axioms and their
study~\cite{shannon_mathematical_1948,jaynes_information_1957,aczel_why_1974,ochs_new_1975,lieb_guide_1998,lieb_fresh_2000,csiszar_axiomatic_2008,baumgartner_characterizing_2012}. While
the choice of a set of axioms is ultimately a matter of taste, we
sketch in the following some of the most popular axioms. We do this
for the case of entropies defined on quantum systems, i.e., we
consider functions $H$ from the set of density operators (denoted by
$\rho$) to the real numbers.
\begin{itemize}
\item \emph{Positivity:} $H(\rho) \geq 0$.
\item \emph{Invariance under isometries:} $H(U \rho U^{\dagger}) =
  H(\rho)$.
\item \emph{Continuity:} $H$ is a continuous function of $\rho$.
\item \emph{Additivity:} $H(\rho_A \otimes \rho_B) = H(\rho_A) +
    H(\rho_B)$.
\item \emph{Subadditivity:} $H(\rho_{A B}) \leq H(\rho_A) +
  H(\rho_B)$.\footnote{Here $\rho_{A B}$ denotes a density operator on
    a bipartite system and $\rho_A$ and $\rho_B$ are obtained by
    partial traces over the second and first subsystem, respectively.}
\end{itemize}
The (special) case of classical entropies is obtained by replacing the
density operators by probability distributions. Note that the second
axiom then reduces to the requirement that the entropy is invariant
under permutations.

It is easy to verify that the von Neumann entropy satisfies the above
axioms. Furthermore, it can be shown that (up to a constant factor,
which may be fixed by an additional normalization axiom) the von
Neumann entropy is essentially the only function satisfying the above
postulates~\cite{ochs_new_1975}.  This result -- as well as similar
results based on slightly different sets of axioms -- nicely expose the
universal nature of entropy. Note, in particular, that the above
axioms do not refer specifically to thermodynamic or
information-theoretic properties of a system.

An alternative to this axiomatic approach is to relate entropy to
\emph{operational} quantities. In thermodynamics, examples for such
operational quantities include measures for heat flow or the amount of
work that is transformed into heat during a given process. In
information theory, operational quantities are, for instance, the
minimum size to which the information generated by a source can be
compressed, or the amount of uniform randomness that can be extracted
from a non-uniform source.

Given the very different nature of these operational quantities, it is
not obvious that this approach can lead to a reasonable notion of
entropy. One would rather expect an entire family of entropy
measures -- possibly as large as the number of different operational
quantities one considers. However, there exist remarkable connections,
even relating thermodynamic and information-theoretic quantities. For
example, it follows from Landauer's
principle~\cite{landauer_irreversibility_1961,bennett_logical_1973}
that the amount of work that can be extracted from a system is
directly related to the size to which the information contained in it
can be compressed~\cite{RARDV11,Dahlsten2011,Faist2012}.

Recent work has shown that a large number of operational quantities
can be characterized with one single class of entropy
measures. \emph{Smooth entropies} (denoted by $H_{\min}^\epsilon$ and
$H_{\max}^\epsilon$), which were developed mostly within quantum
information theory, are an example of such a class. For instance,
$H_{\min}^\epsilon$ quantifies the number of uniformly random
(classical) bits that can be deterministically extracted from a weak
source of randomness\cite{Renner2005,RenKoe2005} and
$H_{\max}^\epsilon$ quantifies the number of bits needed to encode a
given (classical) value\cite{Renes12}. More generally,
$H_{\min}^\epsilon$ can be used to characterize
\emph{decoupling}\cite{Dupuis2009}, a quantum version of randomness
extraction\cite{Dupuis2010}, and \emph{state
  merging}\cite{Horodecki2005,Horodecki2006}, which can be seen as the
fully quantum analogue of coding\cite{Berta2009}. Also, a combination
of $H_{\min}^\epsilon$ and $H_{\max}^\epsilon$ gives an expression for
the classical capacity of a classical\cite{Renner2006} or a
quantum\cite{Renes2011} channel, as well as its ``reverse''
capacity\cite{Berta2011}. Additional applications can be found particularly in quantum cryptography (see, e.g.,
~\cite{Renner2005,Damgard2007,Scarani2008a}). Smooth entropies also
have operational interpretations within thermodynamics. For example,
they can be used in a single-shot version of Landauer's principle to quantify the amount of work required by an
operation that moves a given system into a pure
state\cite{RARDV11,Dahlsten2011,Faist2012}.

However, smooth entropies are generally different from the von Neumann
entropy except in special cases. This implies that many operational
quantities, characterized by smooth entropies, are not in general
accurately described by the von Neumann entropy (e.g. the amount of
extractable randomness or the encoding length). In particular, it follows that some of
the axioms considered above must be incompatible with the operational
approach.

This can also be seen directly, for example, for the
(classical) task of randomness extraction. Let $C(X)$ be the number of uniform bits that can be obtained by applying a function to
a random variable $X$ distributed according to $P_X$. Then the
quantity $C$ automatically has the properties one would expect from an
uncertainty measure: it equals $0$ if $X$ is perfectly known, and it
increases as $X$ becomes more uncertain. One may therefore interpret
$C$ as an (operationally defined) entropy measure for classical random variables.

However, while $C$ is indeed positive, invariant under permutations,
and additive, it is not subadditive. To see this, consider a random
variable $R$ uniformly distributed over the set $\{1, \ldots,
2^\ell\}$, for some large $\ell \in \mathbb{N}$. Furthermore, define
the random variables $X$ and $Y$ by
\begin{align*}
  X & =
  \begin{cases}
    R & \text{if $R \leq 2^{\ell-1}$} \\ 0 & \text{otherwise},
  \end{cases}
  \\
  Y & =
  \begin{cases}
    R & \text{if $R > 2^{\ell-1}$} \\ 0 & \text{otherwise.}
  \end{cases}
\end{align*}
Since $\Pr[X=0] = \Pr[Y=0] = \frac{1}{2}$, it is not possible to
extract more than $1$ bit from either of $X$ or $Y$ separately, i.e.,
$C(X) = C(Y) \leq 1$. However, since the pair $(X,Y)$ is in one-to-one
relation to $R$, we have $C(X Y) = C(R) = \ell$. Hence, subadditivity,
$C(X Y) \leq C(X) + C(Y)$ can be violated by an arbitrarily large
amount.\footnote{However, an inequality of similar form can be recovered --- this is known as the \emph{entropy splitting lemma}~\cite{Wullschleger2007,Damgaard07}.}

\subsection{Generalized entropy measure}

The above considerations show that an operational approach to
entropies necessitates the use of entropy measures that are more
general than those obtained by the usual axiomatic approaches. The aim
of this paper is to investigate such a generalization, which is
motivated by previous
work~\cite{buscemi_quantum_2010,brandao_one-shot_2011,wang_one-shot_2012,tomamichel_hierarchy_2012}.
%{\color{green}
We derive a number of properties of this measure and relate it back to
the better-studied family of smooth entropies.
%}

Our generalized entropy measure is, technically, a family of
entropies, denoted $H_H^\epsilon$, and parametrized by a real number
$\epsilon$ from the interval $[0,1]$. $H_H^\epsilon$ is defined via a
relative-entropy type quantity, i.e., a function that depends on two
density operators, $\rho$ and $\sigma$, similarly to the
Kullback-Leibler divergence~\cite{Kullback1951,Wehrl1978}. This
quantity, denoted $D_H^\epsilon$, has a simple interpretation in the
context of quantum hypothesis testing~\cite{Helstrom1969}. Consider a
measurement for distinguishing whether a system is in state $\rho$ or
$\sigma$. $D_H^\epsilon(\rho\|\sigma)$ then corresponds to the
negative logarithm of the failure probability when the system is in
state $\sigma$, under the constraint that the success probability when
the system is in state $\rho$ is at least $\epsilon$ (see
Section~\ref{sec_relentrdef} below).

Starting from $D_H^\epsilon(\rho\|\sigma)$, it is possible to directly
define a \emph{conditional entropy}, $H_H^\epsilon(A|B)$, i.e., a
measure for the uncertainty of a system $A$ conditioned on a system
$B$ (see Section~\ref{sec_entrdef} below). We note that, while the
conditional von Neumann entropy may be defined analogously using the
Kullback-Leibler divergence, the standard expression for conditional
von Neumann entropy~\cite{Nielsen2000},
\begin{align} 
  \label{eq_condvN} H(A|B) = H(\rho_{A B}) -H(\rho_B)\ ,
\end{align}
cannot be generalized directly. However, as shown in
Section~\ref{sec_chainrule}, $H_H^\epsilon$ satisfies a \emph{chain
  rule}, i.e., an inequality which resembles~\eqref{eq_condvN}. In
addition, we show that $H_H^\epsilon$ has many desirable properties
that one would expect an entropy measure to have (see
Section~\ref{sec_basicproperties}), for instance that it reduces to the von Neumann entropy in the asymptotic limit (Asymptotic Equipartition Property).

%{\color{green}
Apart from deriving the chain rule for the considered entropy measure, the
  main contribution of this paper is to
  %}
%A central part of this contribution is to
establish direct relations
to the \emph{smooth entropy measures} $H_{\min}^\epsilon$ and
$H_{\max}^{\epsilon}$ (Section~\ref{sec_smooth}).  As explained above,
it has been shown that these accurately characterize a number of
operational quantities, such as information compression, randomness
extraction, entanglement manipulation, and channel
coding. Furthermore, they are also relevant in the context of
thermodynamics, e.g., for quantifying the amount of work that can be
extracted from a given system.  The bounds derived in
Section~\ref{sec_smooth} imply that $H_H^\epsilon$ has a similar
operational significance.

\section{Preliminaries}

\subsection{Notation and Definitions}

For a finite-dimensional Hilbert space $\mathcal{H}$, let
$\mathcal{L}(\mathcal{H})$ and $\mathcal{P}(\mathcal{H})$ be the
linear and positive semi-definite operators on $\mathcal{H}$,
respectively. On $\mathcal{L}(\mathcal{H})$ we employ the
Hilbert-Schmidt inner product $\left<X,Y\right>:=\operatorname{Tr}(X^\dagger
Y)$. Quantum states form the set
$\mathcal{S}(\mathcal{H})=\{\rho\in\mathcal{P}(\mathcal{H}):\operatorname{Tr}(\rho)=1\}$,
and we define the set of subnormalized states as
$\mathcal{S}_\leq(\mathcal{H})=\{\rho\in\mathcal{P}(\mathcal{H}):0<\operatorname{Tr}(\rho)\leq1\}$.
To describe multi-partite quantum systems on tensor product spaces we
use capital letters and subscripts to refer to individual subsystems
or marginals. We call a state $\rho_{XB}$ {\it classical-quantum (CQ)}
if it is of the form $\rho_{XB}=\sum_x
p(x)\left|{x}\right>\left<{x}\right|\otimes \rho^x_B$ with
$\rho_B^x\in\mathcal{S}(\mathcal{H}_B)$, $p(x)$ a probability
distribution and $\{\left|{x}\right>\}$ an orthonormal basis of
$\mathcal{H}_X$.

A map $\mathcal{E}:\mathcal{L}(\mathcal{H})\rightarrow
\mathcal{L}(\mathcal{H'})$ for which $\mathcal{E}\otimes\mathcal{I}$, for any $\mathcal{H''}$, maps
$\mathcal{P}(\mathcal{H}\otimes\mathcal{H''})$ to
$\mathcal{P}(\mathcal{H'}\otimes\mathcal{H''})$ is called a completely
positive map (CPM). It is called trace-preserving if
$\operatorname{Tr}(\mathcal{E}[X])=\operatorname{Tr}(X)$ for any
$X\in\mathcal{P}(\mathcal{H})$. A unital map satisfies
$\mathcal{E}(\mathbb{Id})=\mathbb{I}$, and a map is sub-unital if
$\mathcal{E}(\mathbb{I})\leq \mathbb{I}$.  The adjoint $\mathcal{E}^*$
of $\mathcal{E}$ is defined by
$\operatorname{Tr}\left(\mathcal{E}^*(Y)\,X\right) =
\operatorname{Tr}\left(Y\,\mathcal{E}(X)\right)$.

We employ two distance measures on subnormalized states: the purified
distance
$P(\rho,\sigma)$ \cite{gilchrist_distance_2005,rastegin_sine_2006,tomamichel_duality_2010}
and the generalized trace distance $D(\rho,\sigma)=\frac{1}{2}\Vert
\rho-\sigma\Vert_1+\tfrac12|\operatorname{Tr}\rho-\operatorname{Tr}\sigma|$ (where
$\vert\vert\rho\vert\vert_1=\operatorname{Tr}(\sqrt{\rho^\dagger\rho})$).
The purified distance is defined in terms of the generalized fidelity $F(\rho,\sigma)=\Vert\sqrt{\rho}\sqrt{\sigma}\Vert_1+\sqrt{(1-\operatorname{Tr}\rho)(1-\operatorname{Tr}\sigma)}$ by $P(\rho,\sigma)=\sqrt{1-F(\rho,\sigma)^2}$. (The fidelity itself is just the first term in the expression.) The purified and trace distances obey
the following relation~\cite{fuchs_cryptographic_1999}:
$D(\rho,\sigma)\leq P(\rho,\sigma)\leq \sqrt{2D(\rho,\sigma)}$.

Finally, the operator inequalities $A\leq B$ and $A<B$ are taken to mean that $B-A$
is positive semi-definite and positive definite respectively, and when comparing a matrix to a scalar we
assume that the scalar is multiplied by the identity matrix. Note also
that all logarithms taken in the calculations are base 2.

\subsection{Semi-Definite Programs}
\label{sec:sdp}

Watrous has given an elegant formulation of semidefinite programs
especially adapted to the present
context~\cite{watrous_semidefinite_2009}. Here we follow his notation;
see also~\cite{boyd_convex_2004} for a more extensive treatment. A
semidefinite program over $\mathcal{X}=\mathbb{C}^n$ and
$\mathcal{Y}\in\mathbb{C}^m$ is specified by a triple $(\Psi,A,B)$,
for $A$ and $B$ Hermitian operators in $\mathcal{L}(\mathcal{X})$ and
$\mathcal{L}(\mathcal{Y})$ respectively, and
$\Psi:\mathcal{L}(\mathcal{X})\rightarrow \mathcal{L}(\mathcal{Y})$ a
linear, Hermiticity-preserving operation.

This semidefinite program corresponds to two optimization problems, the so-called ``primal'' and ``dual'' problems:\\

\begin{minipage}[t]{0.23\textwidth}
  PRIMAL\\

  minimize\\
  subj. to\\
\end{minipage}
\begin{minipage}[t]{0.23\textwidth}
  \text{}\\
  \\
  $\left<A,X\right>$\\
  $\Psi(X)\geq B$\\
  $X\in\mathcal{P}(\mathcal{X})$

\end{minipage}
\begin{minipage}[t]{0.23\textwidth}
  DUAL\\

  maximize\\
  subj. to\\
\end{minipage}
\begin{minipage}[t]{0.23\textwidth}
  \text{}\\
  \\
  $\left<B,Y\right>$\\
  $\Psi^*(Y)\leq A$\\
  $Y\in\mathcal{P}(\mathcal{Y})$

\end{minipage}\\
\text{}\\
\\
With respect to these problems, one can define the primal and dual
feasible sets $\mathcal{A}$ and $\mathcal{B}$ respectively:
\begin{align}
  \mathcal{A}&=\{X\in\mathcal{P}(\mathcal{X}) :  \Psi(X)\leq B\},\\
  \mathcal{B}&=\{Y\in\mathcal{P}(\mathcal{Y}) : \Psi^*(Y)\geq A\}.
\end{align}
The operators $X\in\mathcal{A}$ and $Y\in\mathcal{B}$ are then called
primal and dual feasible (solutions) respectively.

To each of the primal and dual problems, the associated optimal values
are defined as:\footnote{If $\mathcal{A}=\emptyset$ or
  $\mathcal{B}=\emptyset$, we define $\alpha=\infty$ or
  $\beta=-\infty$ respectively}
\begin{equation*}
  \alpha=\inf_{X\in\mathcal{A}}\left<A,X\right>
  \quad\text{and}\quad\beta=\sup_{Y\in\mathcal{B}}\left<B,Y\right>.
\end{equation*}
Solutions to the primal and dual problems are related by the following
two duality theorems:
\begin{theorem}
  (Weak duality). $\alpha\leq\beta$ for every semidefinite program
  $(\Psi, A, B)$.
\end{theorem}

\begin{theorem}
  (Slater-type condition for strong duality). For every semi-definite
  program $(\Psi, A, B)$ as defined above, the following two
  statements hold:
  \begin{enumerate}
  \item Strict primal feasibility: If $\beta$ is finite and there
    exists an operator $X> 0$ s.t. $\Psi(X)> B$, then $\alpha=\beta$
    and there exists $Y\in\mathcal{B}$ s.t. $\left< B,Y\right>=\beta$.
  \item Strict dual feasibility: If $\alpha$ is finite and there
    exists an operator $Y> 0$ s.t. $\Psi^*(Y)<A$, then $\alpha=\beta$
    and there exists $X\in\mathcal{A}$ s.t. $\left<
      A,X\right>=\alpha$.
  \end{enumerate}
\end{theorem}
Given strict feasibility, we obtain \emph{complementary slackness}
conditions linking the optimal $X$ and $Y$ for the primal and the dual
problem:
\begin{equation}
  \Psi(X)Y=BY\quad\text{and}\quad \Psi^*(Y)X=AX.
\end{equation}

Semidefinite programs can be solved efficiently using the ellipsoid
method~\cite{Grotschel1993}. There exists an algorithm that, under
certain stability conditions and bounds on the primal feasible and
dual feasible sets, finds an approximation for the optimal value of
the primal problem. The running time of the algorithm is bounded by a
polynomial in $n$, $m$, and the logarithm of the desired accuracy
(see~\cite{watrous_semidefinite_2009} for more details).

\section{Relative and Conditional Entropies}
We will now introduce the new family of entropy measures, as well as
the smooth entropies, and the set of relative entropies that they are
based on.

\subsection{Definition of relative entropies} \label{sec_relentrdef}

We define the $\epsilon$-relative entropy $D^{\epsilon} (\rho\vert
\vert \sigma)$ of a subnormalized state $\rho\in\mathcal{S}_\leq(\mathcal{H})$
relative to $\sigma\in\mathcal{P}(\mathcal{H})$ as\footnote{Note that
  this differs slightly from both the definitions used by Wang and
  Renner~\cite{wang_one-shot_2012}, Tomamichel and
  Hayashi~\cite{tomamichel_hierarchy_2012}, and Matthews and Wehner~\cite{matthews_finite_2012}. Similar formulations
  specific to mutual information and entanglement were previously
  given respectively by Buscemi and Datta~\cite{buscemi_quantum_2010}
  and Brand\~ao and Datta~\cite{brandao_one-shot_2011}.\mbox{}}
\begin{equation} \label{eq_Depsdef}
  2^{-D^{\epsilon}(\rho\vert\vert\sigma)}:=\tfrac{1}{\epsilon}\min\{\left<Q,\sigma\right>\vert
  0\leq Q\leq 1 \land \left< Q,\rho\right> \geq \epsilon\} .
\end{equation}
This corresponds to minimizing the probability that a strategy $Q$ to
distinguish $\rho$ from $\sigma$ produces a wrong guess on input
$\sigma$ while maintaining a minimum success probability $\epsilon$ to
correctly identify $\rho$. In particular, for $\epsilon=1$,
$D_H^\epsilon(\rho\vert\vert\sigma)$ is equal to R\'enyi's
entropy\cite{Renyi1961} of order $0$, and $D_0(\rho\vert\vert
\sigma)=-\log\operatorname{Tr}(\rho^0\sigma)$, with $\rho^0$ the projector on
the support of $\rho$~\cite{tomamichel_hierarchy_2012}.

The relative min- and max-entropies $D_{\min}$ and $D_{\max}$
for $\rho\in\mathcal{S}_\leq (\mathcal H)$ and $\sigma\in\mathcal{P}(\mathcal{H})$ are defined as follows:\footnote{The relative max-entropy was introduced in~\cite{datta_min_2009}, but our definition of the relative min-entropy differs from the one used therein.}
\begin{align}
  2^{-D_{\min}(\rho\vert\vert\sigma)}&=\left\|\sqrt{\rho}\sqrt{\sigma}\right\|_1^2\\
  D_{\max}(\rho\vert\vert\sigma)&=\min\{\lambda\in \mathbb{R}:
  2^\lambda\sigma\geq\rho\}.
\end{align}
%where $F(\rho,\sigma)$ denotes the fidelity between the two states
%$\rho$ and $\sigma$, given by
%$F(\rho,\sigma)=\operatorname{Tr}\vert\sqrt\rho\sqrt\sigma\vert$. These
%definitions can easily be extended to subnormalized states
%$\rho,\sigma\in\mathcal{S}_\leqslant(\mathcal{H})$ by using the
%generalized fidelity $F(\rho,\sigma) =
%\operatorname{Tr}\vert\sqrt\rho\sqrt\sigma\vert+\sqrt{(1-\operatorname{Tr}\rho)(1-\operatorname{Tr}\sigma)}$.
We also define the corresponding smoothed quantities:
\begin{align}
  D_{\min}^\epsilon(\rho\vert\vert\sigma)&=\max_{\tilde\rho\in\mathcal{B}_\epsilon(\rho)}D_{\min}(\tilde\rho\vert\vert\sigma),\\
  D_{\max}^\epsilon(\rho\vert\vert\sigma)&=\min_{\tilde\rho\in\mathcal{B}_\epsilon(\rho)}D_{\max}(\tilde\rho\vert\vert\sigma),
\end{align}
with
$\mathcal{B}_\epsilon(\rho)=\{\tilde\rho\in\mathcal{S}_\leq(\mathcal{H})\vert
P(\tilde\rho,\rho)\leq\epsilon\}$ the purified-distance-ball around
$\rho$ so that the optimization is over all subnormalized states
$\tilde\rho$ $\epsilon$-close to $\rho$ with respect to the purified
distance. The latter is given by
$P(\rho,\sigma)=\sqrt{1-F^2(\rho,\sigma)}$.

\subsection{Definition of the conditional
  entropies} \label{sec_entrdef}

We define the new entropy $H_H^\epsilon(A\vert B)_\rho$, in terms of
the relative entropy we have already introduced, as follows:
\begin{align}
  H_H^\epsilon(A\vert
  B)_\rho&:=-D_H^{\epsilon}(\rho_{AB}\vert\vert\mathbb{I}_A\otimes\rho_B)
\end{align}
In the smooth entropy framework, two variants of the min- and max- entropies are given
by:~\cite{tomamichel_duality_2010,Koenig2009IEEE_OpMeaning,TSSR11}
\begin{align}
  {H^\epsilon_{\min}(A\vert B)_{\rho|\sigma}} &:=
  -D_{\max}^\epsilon(\rho_{AB}\Vert\mathbb{I}_A\otimes\sigma_B) ,\\
  H^\epsilon_{\max}\left(A\vert B\right)_{\rho|\sigma} &:=
  {-D_{\min}^\epsilon(\rho_{AB}\Vert\mathbb{I}_A\otimes\sigma_B)}\ ,
\end{align}
\begin{align}
  {H^\epsilon_{\min}(A\vert B)_\rho} &:=
  \max_{\tilde\rho\in\mathcal{B}_\epsilon(\rho)} \max_{\sigma_B\in\mathcal S_\leq (\mathcal H_B)}\;
  -D_{\max}(\tilde\rho_{AB}\Vert\mathbb{I}_A\otimes\sigma_B) ,\\
  H^\epsilon_{\max}\left(A\vert B\right)_\rho &:=
  \min_{\tilde\rho\in\mathcal{B}_\epsilon(\rho)} \max_{\sigma_B\in\mathcal S_\leq (\mathcal H_B)}\;
  {-D_{\min}(\tilde\rho_{AB}\Vert\mathbb{I}_A\otimes\sigma_B)}\ .
\end{align}
The non-smoothed versions $H_{\min}(A|B)$ and $H_{\max}(A|B)$
are given by setting $\epsilon=0$. In both cases, the optimal $\sigma$ is a normalized state, i.e.\ it is sufficient to restrict the maximization to $\sigma_B\in\mathcal S(\mathcal H_B)$.

For the special case when $\epsilon\rightarrow 0$,
$H_H^\epsilon(A\vert B)$ converges to $H_{\min}(A\vert
B)_{\rho|\rho}$ since for the optimal solutions to the semi-definite program as defined below $X\rightarrow 0$. In
the case where one is also not conditioning on any B-system (i.e. take
$B$ to be a trivial system, or take
$\rho_{AB}=\rho_{A}\otimes\rho_{B}$), then $H_H^\epsilon$ reduces to
the min-entropy:
\begin{equation}
  \lim_{\epsilon\rightarrow0}H_H^\epsilon(A)_\rho=H_{\min}(A)_\rho=-\log \vert\vert\rho_A\vert\vert_\infty.
\end{equation}
Note also that $H_H^\epsilon$ is monotonically increasing in
$\epsilon$: to see this, observe that the dual optimal $\{\mu,X\}$ for
$2^{H_H^\epsilon}$ (see below) is also feasible for $2^{H_H^{\epsilon'}}$ with
$\epsilon'\geq\epsilon$.

%
% Properties section
%
\subsection{Elementary Properties} \label{sec_basicproperties} As we
are going to show in this section, the quantities $D_H^\epsilon$ and
$H_H^\epsilon$ we introduced satisfy many desirable properties one
would expect from an entropy measure.

\subsubsection{Properties of $D_H^\epsilon$}

$D_H^\epsilon$ can be expressed in terms of a semi-definite program, meaning it can be efficiently approximated. Due to strong duality we obtain two equivalent expressions with optimal solutions linked by complementary slackness conditions~\cite{boyd_convex_2004}. %watrous doesn't mention slackness in that paper
The semi-definite program for $2^{-D_H^{\epsilon}(\rho\vert\vert\sigma)}$ reads:\\
\\
\begin{minipage}[t]{0.24\textwidth}
  PRIMAL\\

  minimize\\
  subj. to\\
\end{minipage}
\begin{minipage}[t]{0.24\textwidth}
  \text{}\\
  \\
  $\frac{1}{\epsilon}$Tr[Q$\sigma$]\\
  Q$\leq\mathbb{I}$\\
  Tr[Q$\rho$]$\geq\epsilon$\\
  $Q\geq 0$

\end{minipage}
\begin{minipage}[t]{0.24\textwidth}
  DUAL\\

  maximize\\
  subj. to\\
\end{minipage}
\begin{minipage}[t]{0.24\textwidth}
  \text{}\\
  \\
  $\mu-\frac{\text{Tr} [X]}{\epsilon}$\\
  $\mu\rho\leq \sigma + $X\\
  $X\geq 0$\\
$\mu\geq 0$

\end{minipage}

\text{}\\
\\
This yields the following complementary slackness conditions for
primal and dual optimal solutions $\{Q\}$ and $\{ \mu,X \}$:
\begin{align}
  (\mu\rho-X)Q&=\sigma Q\\
  \operatorname{Tr}[Q\rho]&=\epsilon\\
  QX&=X
\end{align}
from which we can infer that $[Q,X]=0$, as well as the fact that the positive part of $(\mu\rho-\sigma)$ is in the eigenspace of $Q$ with eigenvalue 1.\\
\\
Further properties include:%~\cite{wang_one-shot_2012}
\begin{prop}[Positivity]
  For any $\rho, \sigma\in\mathcal{S}(\mathcal{H})$,
  \begin{equation}
    D_H^\epsilon(\rho\vert\vert\sigma)\geq 0,
  \end{equation}
  with equality if $\rho=\sigma$.
\end{prop}
\begin{proof}
  Positivity follows immediately from the definition of $D_H^\epsilon$
  by choosing $Q=\nobreak \epsilon\mathbb{I}$.  Equality is achieved if
  $\rho=\sigma$ because
  $\frac1\epsilon\min_{\operatorname{Tr}(Q\rho)\geq\epsilon}\operatorname{Tr}(Q\rho)=1$.
\end{proof}
Note that $D_H^\epsilon\left(\rho\Vert\sigma\right)=0$ does not
generally imply $\rho=\sigma$: for example, consider the case where
$\epsilon=1$ and where $\rho$ and $\sigma$ have same support.\\

%{\color{green}
The following property relates the hypothesis testing relative entropy to the Trace Distance. Both the proposition and its proof are due to Marco Tomamichel \cite{EmailMarcoTom}.
\begin{prop}[Relation to trace distance]
%\footnote {This proposition and its proof are due to Marco Tomamichel}
For any $\rho, \sigma\in\mathcal{S}(\mathcal{H})$, $0<\epsilon<1$ and $\delta=D(\rho,\sigma)$ the trace distance between $\rho$ and $\sigma$,
\begin{equation}
\log\frac{\epsilon}{\epsilon-(1-\epsilon)\delta}\leq D_H^{\epsilon}(\rho\vert\vert\sigma)\leq \log\frac{\epsilon}{\epsilon-\delta}.
\end{equation}
In particular, we have the Pinsker-like inequality $\frac{1-\epsilon}{\epsilon}\cdot D(\rho,\sigma)\leq D_H^{\epsilon}(\rho\vert\vert\sigma)$. Furthermore, the proposition implies that for $0<\epsilon<1$, $D_H^\epsilon(\rho\vert\vert\sigma)= 0$ if and only if $\rho=\sigma$, inheriting this property from the trace distance.
\end{prop}
\begin{proof}
The trace distance can be written as
\begin{equation}
D(\rho,\sigma)=\max_{0\leq Q\leq 1} \operatorname{Tr}(Q(\rho-\sigma))=\operatorname{Tr}(\{\rho>\sigma\}(\rho-\sigma)),
\end{equation}
where $\{\rho>\sigma\}$ denotes the projector onto the positive part of $(\rho-\sigma)$. We thus immediately have that $\operatorname{Tr}(Q(\rho-\sigma))\leq\delta=D(\rho,\sigma)$ for all $0\leq Q\leq \mathbb{I}$, and so $\operatorname{Tr}(Q\sigma)\geq\operatorname{Tr}(Q\rho)-\delta\geq \epsilon-\delta$ for $Q$ the optimal choice in $D_H^{\epsilon}(\rho\vert\vert\sigma)$. This directly implies that $2^{-D_H^{\epsilon}(\rho\vert\vert\sigma)}\geq \frac{\epsilon-\delta}{\epsilon}$. This proves the upper bound.\\
\\
For the lower bound, we may choose $0\leq \tilde Q\leq \mathbb{I}$ as
\begin{equation}
\tilde Q=(\epsilon-\mu)\mathbb{I}+(1-\epsilon+\mu)\{\rho>\sigma\},  \quad\text{where } \mu=\frac{(1-\epsilon)\operatorname{Tr}(\{\rho>\sigma\}\rho)}{1-\operatorname{Tr}(\{\rho>\sigma\}\rho)}.
\end{equation}
Hence, $\mu=(1-\epsilon+\mu)\operatorname{Tr}(\rho\{\rho>\sigma\})$ and thus
\begin{equation}
\operatorname{Tr}(\tilde Q\rho)=(\epsilon-\mu)+(1-\epsilon+\mu)\operatorname{Tr}(\rho\{\rho>\sigma\})=\epsilon.
\end{equation}
Moreover,
\begin{equation}
\operatorname{Tr}(\tilde Q\sigma)=\epsilon-\mu+(1-\epsilon+\mu)\operatorname{Tr}(\{\rho>\sigma\}\sigma)=\epsilon-\frac{(1-\epsilon)\delta}{1-\operatorname{Tr}(\{\rho>\sigma\}\rho)}\leq \epsilon-(1-\epsilon)\delta.
\end{equation}
Hence, $D_H^{\epsilon}(\rho\vert\vert\sigma)\geq\log\frac{\epsilon}{\epsilon-(1-\epsilon)\delta}$. For the Pinsker-like inequality, observe that $\log\frac{\epsilon}{\epsilon-(1-\epsilon)\delta}=-\log(1-\frac{(1-\epsilon)\delta}{\epsilon})\geq\delta\frac{1-\epsilon}{\epsilon}$.
\end{proof}
%}

\begin{prop}[Data Processing Inequality (DPI)]
  For any completely positive, trace non-increasing map $\mathcal{E}$,
  \begin{equation}
    D_H^\epsilon(\rho\vert\vert\sigma)\geq D_H^\epsilon(\mathcal{E}(\rho)\vert\vert\mathcal{E}(\sigma)).
  \end{equation}
\end{prop}
\begin{proof}
  For a proof of this DPI, see~\cite{wang_one-shot_2012}.
\end{proof}
\begin{prop}[Asymptotic Equipartition Property]
  Let
  \begin{align*}
    D(\rho\vert\vert\sigma)=\operatorname{Tr}[\rho(\log\rho-\log\sigma)]
  \end{align*}
  be the relative entropy between $\rho$ and
  $\sigma$\cite{Wehrl1978}. Then, for any $0<\epsilon < 1$,
  \begin{align}
    \lim_{n\rightarrow \infty}\tfrac{1}{n}\,
    D_H^{\epsilon}(\rho^{\otimes n}\vert\vert\sigma^{\otimes
      n})&=D(\rho\vert\vert\sigma).
  \end{align}
\end{prop}
\begin{proof}
  From Stein's lemma\cite{ogawa_strong_2000,hiai_proper_1991} it
  immediately follows that
  \begin{align}
    \lim_{n\rightarrow \infty}\tfrac{1}{n}\,
    D_H^{\epsilon}(\rho^{\otimes n}\vert\vert\sigma^{\otimes
      n})&=\lim_{n\rightarrow \infty}-\tfrac{1}{n} \log
    \min\tfrac{1}{\epsilon} \operatorname{Tr}\{\sigma^{\otimes n}Q\},
    \\
    &=D(\rho\vert\vert\sigma)-\lim_{n\rightarrow \infty}\tfrac{1}{n}
    \left(\log\tfrac{1}{\epsilon}\right)
    \\
    &=D(\rho\vert\vert\sigma),
  \end{align}
  where the minimum is taken over $0\leq Q \leq 1$ such that
  $\operatorname{Tr}Q\rho\geq \epsilon$.
\end{proof}

\subsubsection{Properties of $H_H^\epsilon$}
\begin{prop}[Bounds]
  \label{prop:bounds}
  For $\rho_{AB}$ an arbitrary normalized quantum state and
  $\rho_{XB}$ a classical-quantum state,
  \begin{align}
    -\log|A|\leq &H_H^\epsilon(A|B)_\rho\leq \log|A|,\\
    0\leq &H_H^\epsilon(X|B)_\rho\leq \log|X|.
  \end{align}
  For classical-quantum states, $H_H^\epsilon(X\vert B)=0$ if $X$ is
  completely determined by $B$ (so that
  $\operatorname{Tr}(\rho_B^x\rho_B^{x'})=0$ for any $x'\neq x$), and the
  entropy is maximal if X is completely mixed and independent of B
  (i.e. $\rho_{XB}=\frac{1}{\vert X \vert}\mathbb{I}_X\otimes\rho_B$).
\end{prop}
\begin{proof}
  Start with the upper bound on $H_H^\epsilon$, and choose
  $\epsilon\mathbb{I}$ as a feasible $Q$:
  \begin{align}
    2^{H_H^\epsilon(A|B)_\rho}
    &=\min_{\operatorname{Tr}[Q_{AB}\rho_{AB}]\geq
      \epsilon}\tfrac1\epsilon\operatorname{Tr}[Q_{AB}\mathbb{I}_A\otimes
    \rho_B]
    \\
    &\leq
    \tfrac1\epsilon\operatorname{Tr}[\epsilon\mathbb{I}_{AB}\mathbb{I}_A\otimes
    \rho_B]
    \\
    &=|A|.
  \end{align}
For the lower bound we use the inequality $|A|\mathbb{I}_A\otimes \rho_B\geq
\rho_{AB}$, which holds for arbitrary quantum states $\rho_{AB}$.
 To establish this inequality, define the superoperator $\mathcal{E}$ as $\mathcal{E}(\rho)=\frac1{d^2}\sum_{j,k}(U^jV^k)\rho(U^jV^k)^{\dagger}$. Here, $d={\rm dim}(\mathcal{H})$ while $U$ and $V$ are unitary operators defined by $\left|{j}\right\rangle=\left|{j+1}\right\rangle$ and $V\left|{k}\right\rangle=\omega^{k}\left|{k}\right\rangle$, for an orthonormal basis set $\{\left|{j}\right\rangle\}_{j=0}^{d-1}$, $\omega=e^{2\pi i/d}$, and where arithmetic inside the ket is taken modulo $d$. (The operators $U$ and $V$ are often called the discrete Weyl-Heisenberg operators, as they generate a discrete projective representation of the Heisenberg algebra.) Then it is easy to work out that $\mathcal{E}\otimes\mathbb{I}[\rho^{AB}]=\frac{1}{|A|}\mathbb{I}_A\otimes \rho_B$, which by the form of $\mathcal{E}$ implies the sought-after inequality. Then, for the optimal $Q_{AB}$ in $H_H^\epsilon(A|B)_\rho$,  
\begin{align}
2^{H_H^\epsilon(A|B)_\rho} &=\frac1\epsilon\operatorname{Tr}[Q_{AB}\,\mathbb{I}_A\otimes\rho_B]\\
&\geq \frac1{\epsilon|A|}\operatorname{Tr}[Q_{AB}\rho_{AB}]\\
&\geq \frac{1}{|A|}.
\end{align}
Classical-quantum states $\rho_{XB}$ obey $\mathbb{I}_X\otimes \rho_B\geq \rho_{XB}$, as $\sum_{x'} p_{x'} \rho^{x'}_B\geq p_x\rho^x_B$ for all $x$. This  implies $H_H^\epsilon(X|B)_\rho\geq  0$ by the same argument.

  That the extremal cases are reached for the described cases follows
  immediately from the respective definitions of $\rho_{XB}$ and
  $H_H^\epsilon$.
\end{proof}

Similarly to $D_H^\epsilon$, $H_H^\epsilon$ also satisfies a data
processing inequality\footnote{This proof is adapted from the DPI
  proof for a differently defined $H^\epsilon$ in Tomamichel and
  Hayashi~\cite{tomamichel_hierarchy_2012}. }.
\begin{prop}[Data Processing Inequality]

  \label{prop:DataProcessingInequalityLea} 
  For any $\rho_{AB}\in\mathcal{S}(\mathcal{H_{AB}})$, let
  $\mathcal{E}:A\rightarrow A'$ be a sub-unital TP-CPM, and
  $\mathcal{F} :B\rightarrow B'$ be a TP-CPM. Then, for
  $\tau_{A'B'}=\mathcal{E}\circ\mathcal{F}(\rho_{AB})$,
  \begin{equation}
    H_H^\epsilon(A\vert B)_\rho \leq H_H^\epsilon(A'\vert B')_\tau
  \end{equation}
\end{prop}
\begin{proof}
  Let $\{\mu,X_{AB}\}$ be dual-optimal for $H_H^\epsilon(A\vert
  B)_\rho$. Starting from $\mu\rho_{AB}\leq
  \mathbb{I}_A\otimes\rho_B+X_{AB}$ and applying
  $\mathcal{E}\circ\mathcal{F}$ to both sides of the inequality
  yields:
  \begin{equation}
    \mu\tau_{AB}\leq \mathcal{E}(\mathbb{I}_A)\otimes\tau_{B'}+\mathcal{E}\circ\mathcal{F}(X_{AB})\leq \mathbb{I}_{A'}\otimes\tau_{B'}+\mathcal{E}\circ\mathcal{F}(X_{AB}).
  \end{equation}
  Hence, $\{\mu,\mathcal{E}\circ\mathcal{F}(X_{AB})\}$ is dual
  feasible for $H_H^\epsilon(A'\vert B')_\tau$ and
  $2^{H_H^\epsilon(A'\vert B')_\tau}\geq
  \mu-\operatorname{Tr}(\mathcal{E}\circ\mathcal{F}(X_{AB})/\epsilon)=2^{H_H^\epsilon(A\vert
    B)_\rho}$.
\end{proof}
\begin{prop}[Asymptotic Equipartition Property]
  For any $0<\epsilon < 1$, it holds that
  \begin{align}
    \lim_{n\rightarrow \infty}\tfrac{1}{n}\, H_H^{\epsilon}(A^{n}\vert
    B^{n})_{\rho^{\otimes n}}&=H(A\vert B)_\rho ,
  \end{align}
  where $H(A\vert B)$ refers to the conditional von Neumann entropy.
\end{prop}
\begin{proof}
  Using the asymptotic property of $D_H^\epsilon$ derived from Stein's
  lemma above, we can show for $H_H^\epsilon(A\vert B)$:
  \begin{align}
    \lim_{n\rightarrow \infty}\tfrac{1}{n} (H_H^{\epsilon}(A^{\otimes
      n}\vert B^{\otimes n})_\rho)&=\lim_{n\rightarrow
      \infty}\tfrac{1}{n} (-D_H^{\epsilon}(\rho^{\otimes
      n}\vert\vert(\mathbb{I}_A\otimes \rho_B)^{\otimes n}))
    \\
    &=-D(\rho_{AB}\vert \vert \mathbb{I}_A\otimes \rho_B)
    \\
    &=-\operatorname{Tr}\rho_{AB}(\log\rho_{AB}-\log\mathbb{I}_A\otimes\rho_B)
    \\
    &=H(AB)-\operatorname{Tr}(\rho_B\log\rho_B)
    \\
    &=H(AB)-H(B)
    \\
    &=H(A\vert B) .
  \end{align}
\end{proof}

\section{Relation to (relative) min- and max-entropies} \label{sec_smooth}

The following propositions relate the new quantities to smooth
entropies. This guarantees an operational significance for
$D_H^\epsilon$ and $H_H^\epsilon$ (see Section~\ref{sec_opapproach}).\footnote{Note that the lower bound on $D_H$ in \eqref{eq:maxvsH} is similar to Lemma 17 of~\cite{datta_strong_2011}.}

%
% D_H <= D_max
%
\begin{prop} {\it Let $\rho\in\mathcal{S}(\mathcal{H_{AB}})$,
    $\sigma\in\mathcal{P}(\mathcal{H_{AB}})$ and
    $0<\epsilon\leq1$. Then,}
  \begin{align}
    D_{\max}^{\sqrt {2\epsilon}}(\rho\vert\vert\sigma)&\leq D_H^{\epsilon}(\rho\vert\vert\sigma)\leq D_{\max}(\rho\vert\vert\sigma)\label{eq:maxvsH}\\
    H_{\min}^{\sqrt{2\epsilon}}(A\vert B)_\rho&\geq
    H_H^\epsilon(A\vert B)_\rho\geq H_{\min}(A\vert
    B)_{\rho\vert\rho}
  \end{align}
\end{prop}
\begin{proof}
  The upper bound for $D_H^\epsilon$ follows immediately from the fact
  that $\mu=2^{-D_{\max}(\rho\vert\vert\sigma)}$ and $X=0$ are
  feasible for $2^{-D_H^\epsilon(\rho\vert\vert\sigma)}$ in the dual
  formulation. For the lower bound, let $\mu$ and X be dual-optimal
  for $2^{-D_H^{\epsilon}(\rho\vert\vert\sigma)}$. Now define
  $G:=\sigma^{1/2}(\sigma+X)^{-1/2}$ and let $\tilde\rho:=G\rho
  G^\dagger$. It thus follows that $\mu \tilde\rho\leq \sigma$, and
  hence $2^{-D_{\max}(\tilde\rho\vert\vert\sigma)}\geq\mu$. Since
  $\operatorname{Tr}[X]\geq 0$, it holds that $\mu\geq
  2^{-D_H^{\epsilon}(\rho\vert\vert\sigma)}$, which implies that
  $2^{-D_H^{\epsilon}(\rho\vert\vert\sigma)}\leq
  2^{-D_{\max}(\tilde\rho\vert\vert\sigma)}$.

  It is now left to prove that the purified distance between
  $\tilde\rho$ and $\rho$ does not exceed $\sqrt{2\epsilon}$: For this
  we employ Lemma~\ref{lem:aeplem}, from which we obtain the upper
  bound $\sqrt{\smash[b]{\frac{2}{\mu}}\operatorname{Tr}[X]}$. Together with $0\leq
  \epsilon\mu-\operatorname{Tr}[X]$, this implies that
  $P(\rho,\tilde\rho)\leq\sqrt{2\epsilon}$, which concludes the proof.

  These bounds can now be rewritten to relate $H_H^\epsilon$ to
  $H_{\min}^\epsilon$. We have
  \begin{equation}
    H_{\min}^{\sqrt{2\epsilon}}(A\vert B)_\rho\geq -D_{\max}^{\sqrt{2\epsilon}}(\rho_{AB}\vert\vert\mathbb{I}_A\otimes\rho_B)\geq -D_H^\epsilon(\rho_{AB}\vert\vert\mathbb{I}_A\otimes \rho_B)= H_H^\epsilon(A\vert B)_\rho.
  \end{equation}
  In the other direction we find:
  \begin{equation}
    H_H^\epsilon(A\vert B)_\rho= -D_H^\epsilon(\rho_{AB}\vert
    \vert\mathbb{I}_A\otimes \rho_B)\geq 
    -D_{\max}(\rho_{AB}\vert\vert\mathbb{I}_A\otimes\rho_B)
    :=H_{\min}(A\vert B)_{\rho\vert\rho}.
  \end{equation}
\end{proof}

%
% D_max <= D_H
%
% D_min <= D_H
%
\begin{prop} {\it Let $\rho\in\mathcal{S}(\mathcal{H})$ and
    $\sigma\in\mathcal{P}(\mathcal{H})$ have intersecting support, and
    $0<\epsilon\leq1$. Then, }
  \begin{gather}
    D_{\min}(\rho\vert\vert\sigma)-\log\frac{1}{\epsilon^2}\leq D_H^{1-\epsilon}(\rho\vert\vert\sigma)\leq D_{\min}^{\sqrt{2\epsilon}}(\rho\vert\vert\sigma)-\log\frac{1}{(1-\epsilon)}
    \\
    H_{\max}(A\vert B)_{\rho}+\log\frac{1}{\epsilon^2}\geq H_H^{(1-\epsilon)}(A\vert B)_\rho
  \end{gather}
\end{prop}
\begin{proof}
  We begin with the lower bound for $D_H^{1-\epsilon}$. Let $\mu$, Q,
  and X be optimal for the primal and dual programs for
  $2^{-D_H^{1-\epsilon}(\rho\vert\vert\sigma)}$ and define
  $Q^\perp:=1-Q$.  Complementary slackness implies
  $\operatorname{Tr}[Q^\perp\rho]=\epsilon$, $QX=X$ and
  $Q(\mu\rho-\sigma-X)=0$. Thus,
  \begin{equation}
    Q(\mu\rho-\sigma-X)=Q(\mu\rho-\sigma)-X,
  \end{equation}
  meaning $Q(\mu\rho-\sigma)$ is hermitian and positive
  semidefinite. This implies that $Q^\perp(\mu\rho-\sigma)$ is also
  hermitian and $Q^\perp(\mu\rho-\sigma)\leq 0$. Since
  $Q+Q^\perp=\mathbb{I}$, this gives a decomposition of
  $(\mu\rho-\sigma)$ into positive and negative parts, and thus
  $\vert\mu\rho-\sigma\vert=Q(\mu\rho-\sigma)-Q^\perp(\mu\rho-\sigma)$. We
  can now proceed:
  \begin{align}
    2^{-\frac{1}{2}D_{\min}(\rho\vert\vert\sigma)}&=\left\|\sqrt{\rho}\sqrt{\sigma}\right\|_1\\
    &=\frac{1}{\sqrt\mu}\left\|\sqrt{\mu\rho}\sqrt{\sigma}\right\|_1\\
    % &\geq\frac{1}{\sqrt\mu}\operatorname{Tr}[\sqrt{\mu\rho}\sqrt\sigma]\\
    &\geq\frac{1}{2\sqrt\mu}\operatorname{Tr}[\mu\rho+\sigma-\vert\mu\rho-\sigma\vert]\\
    &=\frac{1}{2\sqrt\mu}\operatorname{Tr}[\mu\rho+\sigma-Q(\mu\rho-\sigma)+Q^\perp(\mu\rho-\sigma)]\\
    &=\frac{1}{\sqrt{\mu}}\operatorname{Tr}[Q\sigma+\mu Q^\perp\rho]\\
    &\geq \sqrt\mu\operatorname{Tr}[Q^\perp\rho]\\
    &=\sqrt\mu\epsilon\\
    &\geq\epsilon\sqrt{\mu-\operatorname{Tr}[X]/(1-\epsilon)}\\
    &=\epsilon 2^{-\frac{1}{2}D_H^{1-\epsilon}(\rho\vert\vert\sigma)}.
  \end{align}
  We have used that $\vert\vert\sqrt A\sqrt B\vert\vert_1\geq
  \operatorname{Tr}[A+B-\vert A-B\vert]/2$ for positive semidefinite
  A, B (a variation of the trace distance bound on the fidelity; see
  Lemma~A.2.6 of~\cite{Renner2005}).

  Now we prove the upper bound. Let Q be primal-optimal for
  $2^{-D_H^{1-\epsilon}(\rho\vert\vert\sigma)}$, define
  $\tilde\rho:=Q^\frac{1}{2}\rho Q^\frac{1}{2}$, and let $\rho_{AB}$
  be an arbitrary purification of $\rho_A$. Conjugating both sides of
  $\rho_{AB}\leq \mathbb{I}$ by $Q^\frac{1}{2}$, we obtain
  $\tilde\rho_{AB}\leq Q_A\otimes \mathbb{I}_B$.

  The square of the fidelity between two subnormalized states $\zeta$ and $\eta$ can be written also in terms of an SDP, with $\zeta_{AB}$ an arbitrary purification of $\zeta_A$ \cite[Corollary 7]{watrous_semidefinite_2009}:\footnote{Note that this formulation can be brought into the standard form defined in Section~\ref{sec:sdp} by negating the objective functions and interchanging minimization with maximization.}\\

%
%
  % SDP Fidelity:
%
  \begin{minipage}[t] {0.23\textwidth}
    PRIMAL\\

    maximize\\
    subj. to\\
  \end{minipage}
  \begin{minipage}[t] {0.23\textwidth}
    \text{}\\
    \\
    $\text{Tr}[\zeta_{AB}X_{AB}]$\\
    $\text{Tr}_B[X_{AB}]=\eta_A$\\
    $X_{AB}\geq 0$

  \end{minipage}
  \begin{minipage}[t] {0.23\textwidth}
    DUAL\\

    minimize\\
    subj. to\\
  \end{minipage}
  \begin{minipage}[t] {0.23\textwidth}
    \text{}\\
    \\
    $\text{Tr}[Z\eta]$\\
    $\zeta_{AB}\leq Z_A\otimes \mathbb{I}_B$\\
    $Z\geq 0$

  \end{minipage}

  \text{}\\
  \\
  Hence, we see that $Q$ is a feasible $Z_A$ in the SDP for $\left\|\sqrt{\tilde\rho}\sqrt{\sigma}\right\|_1^2$. Hence,
  \begin{align}
    2^{-D_{\min}(\tilde\rho\vert\vert\sigma)}&=\left\|\sqrt{\tilde\rho}\sqrt{\sigma}\right\|_1^2	\\
    &\leq \operatorname{Tr}[Q\sigma]\\
    &=(1-\epsilon)2^{-D_H^{(1-\epsilon)}(\rho\vert\vert\sigma)},
  \end{align}
  and so $D_{\min}(\tilde\rho\vert\vert\sigma)\geq
  D_H^{(1-\epsilon)}(\rho\vert\vert\sigma)+\log\frac{1}{1-\epsilon}$.

  From complementary slackness we get that
  $\operatorname{Tr}[Q\rho]=1-\epsilon$. Using Lemma~\ref{lem:unclem}
  we obtain $P(\tilde\rho,\rho)\leq\sqrt{1-\operatorname{Tr}[Q\rho]^2}
  \leq\sqrt{2\epsilon}$, and the first part of the proposition
  follows.\\
\\
  % D_H <= D_min
%
%

  Rewriting this for $H_{\max}$ and $H_H^{(1-\epsilon)}$
  yields:
  \begin{align}
    H_{\max}(A\vert B)_{\rho}&\geq H_{\max}(A\vert B)_{\rho\vert\rho}\\
    &= -D_{\min}(\rho_{AB}\vert\vert\mathbb{I}_A\otimes\rho_{B})\\
    &\geq -D_H^{1-\epsilon}(\rho_{AB}\vert\vert\mathbb{I}_A\otimes\rho_{B})-\log\frac{1}{\epsilon^2}\\
    &= H_H^{(1-\epsilon)}(A\vert B)_\rho-\log\frac{1}{\epsilon^2}
    % &\geq
    % -D_{\min}^{\sqrt{2\epsilon}}(\rho\vert\vert\rho)-\log\frac{1}{\epsilon^2}+\log\frac{1}{1-\epsilon}
  \end{align}
\end{proof}

\section{Decomposition of Hypothesis Tests \& Entropic Chain Rules}
\label{sec:decomp}
\label{sec_chainrule}
In this section we prove a bound on hypothesis testing between arbitrary states $\rho$ and states $\sigma$ invariant under a group action, in terms of hypothesis tests between $\rho$ and its group symmetrized version $\xi$ and $\xi$ and $\sigma$. This bound yields a chain rule for the hypothesis testing entropy. 
 For a group $G$ and unitary representation $U_g$, let $\mathcal{E}_G(\rho)=\frac1{|G|}\sum_{g\in G}U_g\rho U^\dagger_g$, which is a quantum operation.  (For simplicity of presentation we assume the group is finite, but the argument applies to continuous groups as well.)

\begin{prop}
For any $\rho,\sigma\in\mathcal{S}(\mathcal{H})$ and group $G$ such that $\sigma=\mathcal{E}_G(\sigma)$, let $\xi=\mathcal{E}_G(\rho)$. Then, for $\epsilon,\epsilon'> 0$,
\begin{align}
D_H^{\epsilon+\sqrt{2\epsilon'}}(\rho||\sigma)\leq D_H^{\epsilon}(\rho||\xi)+D_H^{\epsilon'}(\xi||\sigma)+\log\frac{\epsilon+\sqrt{2\epsilon'}}{\epsilon}.
\end{align}
\end{prop}
\begin{proof}
Let $\mu_1$ and $X_1$ be optimal in the dual program of $D_H^{\epsilon}(\rho||\xi)$ and, similarly, $\mu_2$ and $X_2$ be optimal in $D_H^{\epsilon'}(\xi||\sigma)$. Thus,
$\mu_1\rho\leq \xi+X_1$ and $\mu_2\xi\leq \sigma+X_2$. Observe that $X_2$ can be chosen $G$-invariant without loss of generality, since $\mu_2\xi\leq \sigma+\mathcal{E}_G(X_2)$ and $\operatorname{Tr}[X_2]=\operatorname{Tr}[\mathcal{E}_G(X_2)]$.

Chaining the inequalities gives
\begin{align}
\mu_1\mu_2\rho\leq \sigma+X_2+\mu_2 X_1.
\end{align}
Next, define $T=\sigma^{\frac12}(\sigma+X_2)^{-\frac12}$ and conjugate both sides of the above by $T$. This gives 
\begin{align}
\mu_1\mu_2T\rho T^\dagger \leq \sigma+\mu_2 TX_1T^\dagger.
\end{align}
Thus, the pair $\mu_1\mu_2$, $\mu_2TX_1T^\dagger$ is feasible for $D_H^\epsilon(T\rho T^\dagger||\sigma)$. 
Since $T$ is a contraction ($TT^\dagger\leq \mathbb{I}$), we can proceed as follows:
\begin{align}
2^{-D_H^\epsilon(T\rho T^\dagger \vert\vert\sigma)}&\geq\mu_1\mu_2-\frac{\mu_2\operatorname{Tr}[TX_1T^\dagger]}{\epsilon}\\
&\geq \mu_1\mu_2-\frac{\mu_2\operatorname{Tr} X_1}{\epsilon}\\
&=\mu_2 2^{-D_H^\epsilon(\rho||\xi)}\\
&\geq 2^{-D_H^{\epsilon'}(\xi||\sigma)}2^{-D_H^\varepsilon(\rho||\xi)}.
\label{eq:halfchainrule}
\end{align}
Now we show that $P(\rho,T\rho T^\dagger)\leq \sqrt{2\epsilon'}$, in order to invoke Lemma~\ref{lem:smoothing}. Let the isometry $V:{\mathcal{H}_A\rightarrow \mathcal{H}_A\otimes\mathcal{H}_R}$ be a Stinespring dilation of $\mathcal{E}_G$, so that
$\overline{\xi}_{AR}=V_{A\rightarrow AR}\rho_A V_{A\rightarrow AR}^\dagger=\frac1{|G|}\sum_{g,g'\in G}U_g \rho U^\dagger_{g'}\otimes\left|{g}\right\rangle\left\langle{g'}\right|$. 
The state $\overline{\xi}_{AR}$ is an extension of $\xi_A$ since $\xi_A=\operatorname{Tr}_R[\overline{\xi}_{AR}]$. Clearly $T_A\overline{\xi}_{AR}T_A^\dagger$ is an extension of $T\xi T^\dagger$. We now apply Lemma~\ref{lem:aeplem} to the inequality $\xi\leq \sigma/\mu_2+X_2/\mu_2$, noting that the contraction in the lemma is just the operator $T$, to find 
\begin{align}
	P(\bar{\xi}_{AR}, T_A \bar{\xi}_{AR} T_A^{\dagger}) &\leqslant \sqrt{\frac{\operatorname{Tr}[X_2]}{\mu_2}\left( 2 - \frac{\operatorname{Tr}[X_2]}{\mu_2} \right)}\\
	&\leq \sqrt{2 \epsilon'}.
\end{align}

This entails that
\begin{align}
	P(\rho, T\rho T^{\dagger}) &= P(V \rho_A V^{\dagger}, V T \rho T^{\dagger} V^{\dagger})\\
	&= P(V \rho_A V^{\dagger}, T V \rho V^{\dagger}T^{\dagger} )\\
	&= P(\bar{\xi}_{AR}, T_A \bar{\xi}_{AR} T_A^{\dagger})\\
	&\leqslant \sqrt{2 \epsilon'},
\end{align}
where we have used the fact that $T_A$ commutes with $V_{AR}$. This then implies that $\tfrac12||\rho - T\rho T^\dagger ||_1\leq \sqrt{2\epsilon'}$. Lemma~\ref{lem:smoothing} and (\ref{eq:halfchainrule}) then yields the proposition: 
\begin{align}
D_H^{\epsilon+\sqrt{2\epsilon'}}(\rho||\sigma)+\log\frac{\epsilon}{\epsilon+\sqrt{2\epsilon'}}&\leq
D_H^\epsilon(T\rho T^\dagger ||\sigma)\\
&\leq D_H^\epsilon(\rho||\xi)+D_H^{\epsilon'}(\xi||\sigma).
\end{align}
\end{proof}

\begin{corollary}[Chain rule for $H_H^\epsilon$]
 {\it Let $\rho_{ABC}\in\mathcal{S}(\mathcal{H})$ be an arbitrary normalized state, and $\epsilon,\epsilon'>0$. Then, }
\begin{equation}
H_H^{\epsilon+\sqrt{8\epsilon'}}(AB\vert C)_\rho\geq H^\epsilon(A\vert BC)_\rho+H^{\epsilon'}(B\vert C)_\rho-\log\frac{\epsilon+\sqrt{2\epsilon'}}{\epsilon}.
\end{equation}
\end{corollary}
\begin{proof}
Let $G$ be the Weyl-Heisenberg group representation (as in the proof of Prop~\ref{prop:bounds}) acting on $A$, for which $\mathcal{E}_G(\rho_{ABC})=\pi_A\otimes \rho_{BC}$, where $\pi_A=\mathbb{I}/{\rm dim}(\mathcal{H}_A)$. Applied to the hypothesis test between $\rho_{ABC}$ and $\pi_{AB}\otimes \rho_C$, we find
\begin{align}
&\!\!D_H^{\epsilon+\sqrt{8\epsilon'}}(\rho_{ABC}||\pi_{AB}\otimes\rho_C)\nonumber\\
&\leq D_H^{\epsilon}(\rho_{ABC}||\pi_A\otimes\rho_{BC})+D_H^{\epsilon'}(\pi_A\otimes \rho_{BC}||\pi_{AB}\otimes\rho_C)+\log\frac{\epsilon+\sqrt{2\epsilon'}}{\epsilon}\\
&\leq D_H^{\epsilon}(\rho_{ABC}||\pi_A\otimes\rho_{BC})+D_H^{\epsilon'}( \rho_{BC}||\pi_{B}\otimes\rho_C)+\log\frac{\epsilon+\sqrt{2\epsilon'}}{\epsilon}.
\end{align}
As $H_H^\epsilon(A|B)_\sigma=\log d_A-D_H^\epsilon(\sigma_{AB}||\pi_A\otimes \sigma_B)$, this is equivalent to the desired result.
\end{proof}

\section*{Acknowledgements}
We acknowledge discussions with Marco Tomamichel. Research leading to
these results was supported by the Swiss National Science Foundation
(through the National Centre of Competence in Research `Quantum
Science and Technology' and grant No. 200020-135048) and the European
Research Council (grant No. 258932).

\appendix

\section{Useful Lemmas}
 
\begin{lemma}
\label{lem:distancesdp}	
	For $\rho,\sigma\in\mathcal S_\leq(\mathcal H)$,% and $\sigma\in\mathcal S_\leq(\mathcal H)$, 
	\begin{align}
		\max_{0\leq P\leq \mathbb{I}}{\rm
                  Tr}[P(\rho-\sigma)]=D(\rho,\sigma) \ . %\tfrac12||\rho-\sigma||_1+\tfrac12(1-{\rm Tr}\,\sigma).
	\end{align}
\end{lemma}
\begin{proof}
	The proof proceeds, as in \cite[9.22]{Nielsen2000}, by showing the lefthand side is both bounded below and above by the righthand side. 
	Suppose ${\rm Tr}\rho\geq {\rm Tr}\sigma$, otherwise interchange the states. Since $\rho-\sigma$ is Hermitian, we may write $\rho-\sigma=A-B$ for $A=\{\rho-\sigma\}_+$, the positive part of $\rho-\sigma$ and $B=\{\rho-\sigma\}_-$ the nonpositive part. Since $A$ and $B$ have disjoint supports, we have $\left\|\rho-\sigma\right\|_1={\rm Tr}A+{\rm Tr}B$ and ${\rm Tr}A-{\rm Tr}B={\rm Tr}\rho-{\rm Tr}\,\sigma=|{\rm Tr}\rho-{\rm Tr}\,\sigma|$. 
 Then, for $Q$ the projector onto the support of $A$, 
 \begin{align}
 	{\rm Tr}[Q(\rho-\sigma)]&={\rm Tr}[Q(A-B)]\\
 	&={\rm Tr}[A]\\
 	&=\tfrac12||\rho-\sigma||_1+\tfrac12\left|{\rm Tr}\rho-{\rm Tr}\,\sigma\right|.
 \end{align}
 Since $Q$ is a feasible $P$ in the statement of the lemma, this establishes the lower bound. The upper bound follows since, for any feasible $P$, 
 \begin{align}
 	{\rm Tr}[P(\rho-\sigma)]&={\rm Tr}[P(A-B)]\\
 	&\leq {\rm Tr}[PA]\\
 	&\leq {\rm Tr}[A],
 \end{align}
 which is the upper bound. 
\end{proof}

\begin{lemma}
\label{lem:smoothing} 
Let $\rho,\tilde\rho\in\mathcal S_\leq(\mathcal{H})$ be such that $D(\rho,\tilde\rho)\leq \delta$ %$\frac12\vert\vert \tilde\rho-\rho\vert\vert_1\leq\delta$ 
for some $\delta\geq 0$. Then, for any $\sigma\in\mathcal{P}(\mathcal{H})$,
\begin{equation}
D_H^{\epsilon+\delta}(\rho\vert\vert\sigma)+\log\frac{\epsilon}{\epsilon+\delta}\leq D_H^\epsilon(\tilde\rho\vert\vert\sigma).
\end{equation}
\end{lemma}
\begin{proof}
	
%First observe that $P(\rho,\tilde\rho)\leq \delta$ implies ${\rm Tr}[\tilde\rho]\geq 1-\delta^2$. This can be seen as follows. When ${\rm Tr}[\rho]=1$, $ $P(\rho,\tilde\rho)^2=1-\left\|\sqrt{\rho}\sqrt{\tilde\rho}\right\|_1^2$, and therefore $\left\|\sqrt{\rho}\sqrt{\tilde\rho}\right\|_1^2\geq 1-\delta^2$. But, from the Schwarz inequality for the trace norm \cite[IX.32]{BookBhatiaMatrixAnalysis1997}, $\left\|\sqrt{\rho}\sqrt{\tilde\rho}\right\|_1^2\leq \left\|\rho\right\|_1\cdot\left\|\tilde\rho\right\|_1={\rm Tr}\tilde\rho$, giving the desired inequality. 

Let $Q$ be primal-optimal for $D_H^{\epsilon+\delta}(\rho\vert\vert\sigma)$. It follows from Lemma~\ref{lem:distancesdp} that
\begin{align}
  \delta&\geq \max_{0\leq P\leq \mathbb{I}}\operatorname{Tr}[P(\rho-\tilde\rho)]\\
  &\geq \operatorname{Tr}[Q\rho]-\operatorname{Tr}[Q\tilde\rho]\\
  &=\epsilon+\delta-\operatorname{Tr}[Q\tilde\rho]
\end{align}
Hence, $\operatorname{Tr}[Q\tilde\rho]\geq \epsilon$ and $Q$ is primal-feasible for $D_H^\epsilon(\tilde\rho\vert\vert\sigma)$, yielding a bound of
\begin{align}
2^{-D_H^\epsilon(\tilde\rho\vert\vert\sigma)}&\leq\frac{1}{\epsilon}\operatorname{Tr}[Q\sigma]\\
&=\frac{\epsilon+\delta}{\epsilon}2^{-D_H^{\epsilon+\delta}(\rho\vert\vert\sigma)},
\end{align}
which proves the lemma.
\end{proof}

\begin{lemma}[Lemma 7, Berta \emph{et
    al.}\cite{berta_uncertainty_2010}]
  \label{lem:unclem}  For any
    $\rho\in\mathcal{S}_\leq(\mathcal{H})$, and for any nonnegative
    operator $\Pi\leq\mathbb{I}$,
  \begin{equation}
    P(\rho,\Pi\rho\Pi)\leq\frac{1}{\sqrt{\operatorname{Tr}\rho}}\sqrt{(\operatorname{Tr}\rho)^2-(\operatorname{Tr}(\Pi^2\rho))^2}
  \end{equation}
\end{lemma}
\begin{proof}
  Since
  $\vert\vert\sqrt{\rho}\sqrt{\Pi\rho\Pi}\vert\vert_1=\operatorname{Tr}\sqrt{(\sqrt{\rho}\Pi\sqrt{\rho})(\sqrt{\rho}\Pi\sqrt{\rho})}=\operatorname{Tr}(\Pi\rho)$,
  we can write the generalized fidelity as
  \begin{equation}
    \bar F(\rho,\Pi\rho\Pi)=\operatorname{Tr}(\Pi\rho)+\sqrt{(1-\operatorname{Tr}\rho)(1-\operatorname{Tr}(\Pi^2\rho))}.
  \end{equation}
  For simplicity, introduce the following abbreviations:
  $r=\operatorname{Tr}\rho$, $s=\operatorname{Tr}(\Pi\rho)$ and
  $t=\operatorname{Tr}(\Pi^2\rho)$. As $\rho\leq 1$ and $\Pi\leq 1$ trivially
  $0\leq t\leq s \leq r \leq 1$. In terms of these variables, we now
  have that
  \begin{equation}
    1-\bar F(\rho,\Pi\rho\Pi)^2 = r+t-rt-s^2-2s\sqrt{(1-r)(1-t)}.
  \end{equation}
  Since $P(\rho,\Pi\rho\Pi)=\sqrt{1-\bar F(\rho,\Pi\rho\Pi)^2}$, it is
  sufficient to show that $r(1-\bar F(\rho,\Pi\rho\Pi)^2)-r^2+t^2\leq
  0$. This we can establish:
  \begin{align}
    r(1-\bar F(\rho,\Pi\rho\Pi)^2)-r^2+t^2&=r(r+t-rt-s^2-2s\sqrt{(1-r)(1-t)})-r^2+t^2\\
    &\leq r(r+t-rt-s^2-2s(1-r))-r^2+t^2\\
    &=rt-r^2t+t^2-2rs+2r^2s-rs^2\\
    &\leq rt-r^2t+t^2-2rs+2r^2s-rt^2\\
    &=(1-r)(t^2+rt-2rs)\\
    &\leq (1-r)(s^2+rs-2rs)\\
    &=(1-r)s(s-r)\\
    &\leq 0
  \end{align}
  and the lemma follows.
\end{proof}

\begin{lemma}[Lemma 15, Tomamichel \emph{et
    al.}\cite{tomamichel_fully_2009}; Lemma 6.1~\cite{tomamichel_framework_2012}]
  \label{lem:aeplem}
  {\it Let $\rho\in\mathcal{S}(\mathcal{H})$, $\sigma\in\mathcal{P}(\mathcal{H})$, $\rho\leq \sigma+\Delta$, and $G:=\sigma^\frac{1}{2}(\sigma+\Delta)^{-\frac{1}{2}}$, where the inverse is taken on the support of $\sigma$. Furthermore, let $\left|{\psi}\right\rangle \in \mathcal{S}(\mathcal{H} \otimes \mathcal{H})$ be a purification of $\rho$. Then,} 
\begin{equation}
P(\psi, (G \otimes \mathbb{I}) \psi (G^\dagger \otimes \mathbb{I}))\leq\sqrt{\operatorname{Tr} \Delta(2- \operatorname{Tr} \Delta)}.
\end{equation}
\end{lemma}

\begin{proof}
  Let $\left|{\psi}\right>\in \mathcal{S}(\mathcal{H}\otimes\mathcal{H})$ be a purification of $\rho$. Then,
  $(G\otimes\mathbb{I})\left|{\psi}\right>$ is a
  purification of $G\rho G^\dagger$, and with the help of Uhlmann's
  theorem we can bound the fidelity:
  \begin{align}
F(\psi,(G \otimes \mathbb{I}) \psi (G^\dagger \otimes \mathbb{I})) &= \vert\left<{\psi}\right| G\otimes\mathbb{I}\left|{\psi}\right>\vert\\
&\geq\mathcal{R}\{\operatorname{Tr}(G\rho)\}=\operatorname{Tr}(\bar G\rho),
  \end{align}
  with $\bar G := \frac{1}{2}(G+G^\dagger)$. %It thus follows that
%  \begin{equation}
%    P(\rho_{AB},G\rho G^\dagger)\leq\sqrt{(1+\operatorname{Tr}(\bar G\rho_{AB}))(1-\operatorname{Tr}(\bar G \rho_{AB}))}.
%  \end{equation}
  Since G is a contraction\footnote{One can see this by conjugating both sides
    of $\sigma\leq\sigma+\Delta$ by $(\sigma+\Delta)^{-1/2}$, which gives
    $G^\dagger G\leq 1$}, $\vert\vert G\vert\vert\leq 1$. Also,
  $\vert\vert \bar G\vert\vert\leq 1$ by the triangle inequality and
  thus $\operatorname{Tr}(\bar G\rho_{AB})\leq 1$. Furthermore,
  \begin{align}
1-\operatorname{Tr}(\bar G \rho)&=\operatorname{Tr}((\mathbb{I}-\bar G)\rho)\\
&\leq \operatorname{Tr}(\sigma+\Delta)-\operatorname{Tr}(\bar G(\sigma+\Delta))\\
&=\operatorname{Tr}(\sigma+\Delta)-\operatorname{Tr}((\sigma+\Delta)^\frac{1}{2}(\sigma)^\frac{1}{2})\\
&\leq\operatorname{Tr}(\Delta),
  \end{align}
where we have used $\rho\leq\sigma+\Delta$ and $\sqrt{\sigma+\Delta}\geq\sqrt{\sigma}$. Then we find 
\begin{align}
P(\psi, (G \otimes \mathbb{I}) \psi (G^\dagger \otimes \mathbb{I}))&=\sqrt{1-F(\psi, (G \otimes \mathbb{I}) \psi (G^\dagger \otimes \mathbb{I}))^2}\\
&\leq \sqrt{1-(1-\operatorname{Tr}(\Delta)^2)}\\
&=\sqrt{\operatorname{Tr}\Delta(2-\operatorname{Tr}\Delta)}.
\end{align}
\end{proof}

\bibliographystyle{ieeetr}%ws-procs975x65}
\bibliography{hypotesting}

\begin{thebibliography}{10}

\bibitem{shannon_mathematical_1948}
C.~E. Shannon, ``A mathematical theory of communication,'' {\em Bell System
  Technical Journal}, vol.~27, no.~3, pp.~379--423, 1948.

\bibitem{neumann_mathematical_1996}
J.~v. Neumann, {\em Mathematical Foundations of Quantum Mechanics}.
\newblock Princeton University Press, 1996.

\bibitem{ogawa_strong_2000}
T.~Ogawa and H.~Nagaoka, ``Strong converse and {S}tein's lemma in quantum
  hypothesis testing,'' {\em IEEE Transactions on Information Theory}, vol.~46,
  no.~7, pp.~2428 --2433, 2000.

\bibitem{Nagaoka2007}
H.~Nagaoka and M.~Hayashi, ``An information-spectrum approach to classical and
  quantum hypothesis testing for simple hypotheses,'' {\em IEEE Transactions on
  Information Theory}, vol.~53, no.~2, pp.~534 --549, 2007.

\bibitem{han_information-spectrum_2002}
T.~S. Han, {\em Information-Spectrum Method in Information Theory}.
\newblock Springer-Verlag, 2002.

\bibitem{Bowen2006}
G.~Bowen and N.~Datta, ``Beyond i.i.d. in quantum information theory,'' in {\em
  2006 IEEE International Symposium on Information Theory}, pp.~451--455, 2006.

\bibitem{Renner2004}
R.~Renner and S.~Wolf, ``Smooth {R}enyi entropy and applications,'' in {\em
  2004 IEEE International Symposium on Information Theory}, pp.~232--232, IEEE,
  2004.

\bibitem{Renner2005}
R.~Renner, {\em Security of Quantum Key Distribution}.
\newblock PhD thesis, ETH Zurich, 2005.
\newblock arXiv:quant-ph/0512258.

\bibitem{DatRen09}
N.~Datta and R.~Renner, ``Smooth entropies and the quantum information
  spectrum,'' {\em IEEE Transactions on Information Theory}, vol.~55, no.~6,
  pp.~2807--2815, 2009.

\bibitem{jaynes_information_1957}
E.~T. Jaynes, ``Information theory and statistical mechanics,'' {\em Physical
  Review}, vol.~106, no.~4, p.~620, 1957.

\bibitem{aczel_why_1974}
J.~Aczél, B.~Forte, and C.~T. Ng, ``Why the {S}hannon and {H}artley entropies
  are {'natural'},'' {\em Advances in Applied Probability}, vol.~6, no.~1,
  p.~131, 1974.

\bibitem{ochs_new_1975}
W.~Ochs, ``A new axiomatic characterization of the von {N}eumann entropy,''
  {\em Reports on Mathematical Physics}, vol.~8, no.~1, pp.~109--120, 1975.

\bibitem{lieb_guide_1998}
E.~H. Lieb and J.~Yngvason, ``A guide to entropy and the second law of
  thermodynamics,'' {\em Notices of the American Mathematical Society},
  vol.~45, no.~5, p.~571, 1998.

\bibitem{lieb_fresh_2000}
E.~H. Lieb and J.~Yngvason, ``A fresh look at entropy and the second law of
  thermodynamics,'' {\em Physics Today}, vol.~53, no.~4, pp.~32--37, 2000.

\bibitem{csiszar_axiomatic_2008}
I.~Csiszár, ``Axiomatic characterizations of information measures,'' {\em
  Entropy}, vol.~10, no.~3, pp.~261--273, 2008.

\bibitem{baumgartner_characterizing_2012}
B.~Baumgartner, ``Characterizing entropy in statistical physics and in quantum
  information theory.'' arXiv:1206.5727, 2012.

\bibitem{landauer_irreversibility_1961}
R.~Landauer, ``Irreversibility and heat generation in the computing process,''
  {\em {IBM} Journal of Research and Development}, vol.~5, no.~3, p.~183, 1961.

\bibitem{bennett_logical_1973}
C.~Bennett, ``Logical reversibility of computation,'' {\em {IBM} Journal of
  Research and Development}, vol.~17, no.~6, p.~525, 1973.

\bibitem{RARDV11}
L.~del Rio, J.~Åberg, R.~Renner, O.~C.~O. Dahlsten, and V.~Vedral, ``The
  thermodynamic meaning of negative entropy,'' {\em Nature}, vol.~474,
  no.~7349, pp.~61--63, 2011.

\bibitem{Dahlsten2011}
O.~C.~O. Dahlsten, R.~Renner, E.~Rieper, and V.~Vedral, ``Inadequacy of von
  {N}eumann entropy for characterizing extractable work,'' {\em New Journal of
  Physics}, vol.~13, no.~5, p.~053015, 2011.

\bibitem{Faist2012}
P.~Faist, F.~Dupuis, J.~Oppenheim, and R.~Renner, ``A quantitative {L}andauer's
  principle.'' arXiv:1211.1037, 2012.

\bibitem{RenKoe2005}
R.~Renner and R.~K\"{o}nig, ``Universally composable privacy amplification
  against quantum adversaries,'' in {\em Theory of Cryptography}, vol.~3378 of
  {\em Lecture Notes in Computer Science}, pp.~407--425, Springer, 2005.

\bibitem{Renes12}
J.~M. Renes and R.~Renner, ``One-shot classical data compression with quantum
  side information and the distillation of common randomness or secret keys,''
  {\em IEEE Transactions on Information Theory}, vol.~58, pp.~1985--1991, 2012.

\bibitem{Dupuis2009}
F.~Dupuis, {\em The Decoupling Approach to Quantum Information Theory}.
\newblock PhD thesis, Universit\'e de Montr\'eal, 2009.
\newblock arXiv:1004.1641.

\bibitem{Dupuis2010}
F.~Dupuis, M.~Berta, J.~Wullschleger, and R.~Renner, ``The decoupling
  theorem.'' arXiv:1012.6044, 2010.

\bibitem{Horodecki2005}
M.~Horodecki, J.~Oppenheim, and A.~Winter, ``Partial quantum information.,''
  {\em Nature}, vol.~436, no.~7051, pp.~673--6, 2005.

\bibitem{Horodecki2006}
M.~Horodecki, J.~Oppenheim, and A.~Winter, ``Quantum state merging and negative
  information,'' {\em Communications in Mathematical Physics}, vol.~269, no.~1,
  pp.~107--136, 2006.

\bibitem{Berta2009}
M.~Berta, ``Single-shot quantum state merging,'' Master's thesis, ETH Zurich,
  2008.
\newblock arXiv:0912.4495.

\bibitem{Renner2006}
R.~Renner, S.~Wolf, and J.~Wullschleger, ``The single-serving channel
  capacity,'' in {\em 2006 IEEE International Symposium on Information Theory},
  pp.~1424--1427, IEEE, 2006.

\bibitem{Renes2011}
J.~M. Renes and R.~Renner, ``Noisy channel coding via privacy amplification and
  information reconciliation,'' {\em IEEE Transactions on Information Theory},
  vol.~57, no.~11, pp.~7377--7385, 2011.

\bibitem{Berta2011}
M.~Berta, M.~Christandl, and R.~Renner, ``The quantum reverse {S}hannon theorem
  based on one-shot information theory,'' {\em Communications in Mathematical
  Physics}, vol.~306, no.~3, pp.~579--615, 2011.

\bibitem{Damgard2007}
I.~B. {Damg\aa{}rd}, S.~Fehr, R.~Renner, L.~Salvail, and C.~Schaffner, ``A
  tight high-order entropic quantum uncertainty relation with applications,''
  in {\em Advances in Cryptology - CRYPTO 2007}, vol.~4622, pp.~360--378,
  Springer, 2007.

\bibitem{Scarani2008a}
V.~Scarani and R.~Renner, ``Quantum cryptography with finite resources:
  Unconditional security bound for discrete-variable protocols with one-way
  postprocessing,'' {\em Physical Review Letters}, vol.~100, no.~20, pp.~1--4,
  2008.

\bibitem{Wullschleger2007}
J.~Wullschleger, ``Oblivious-transfer amplification,'' in {\em Advances in
  Cryptology---EUROCRYPT '07}, vol.~4515 of {\em Lecture Notes in Computer
  Science}, pp.~555--572, Springer, 2007.

\bibitem{Damgaard07}
I.~B. {Damg\aa{}rd}, S.~Fehr, R.~Renner, L.~Salvail, and C.~Schaffner, ``A
  tight high-order entropic quantum uncertainty relation with applications,''
  in {\em Advances in Cryptology---CRYPTO '07}, vol.~4622 of {\em Lecture Notes
  in Computer Science}, pp.~360--378, Springer, 2007.

\bibitem{buscemi_quantum_2010}
F.~Buscemi and N.~Datta, ``The quantum capacity of channels with arbitrarily
  correlated noise,'' {\em {IEEE} Transactions on Information Theory}, vol.~56,
  no.~3, pp.~1447--1460, 2010.

\bibitem{brandao_one-shot_2011}
F.~Brandão and N.~Datta, ``One-shot rates for entanglement manipulation under
  non-entangling maps,'' {\em {IEEE} Transactions on Information Theory},
  vol.~57, no.~3, pp.~1754--1760, 2011.

\bibitem{wang_one-shot_2012}
L.~Wang and R.~Renner, ``One-shot classical-quantum capacity and hypothesis
  testing,'' {\em Physical Review Letters}, vol.~108, no.~20, p.~200501, 2012.

\bibitem{tomamichel_hierarchy_2012}
M.~Tomamichel and M.~Hayashi, ``A hierarchy of information quantities for
  finite block length analysis of quantum tasks.'' {arXiv:1208.1478}, 2012.

\bibitem{Kullback1951}
S.~Kullback and R.~Leibler, ``On information and sufficiency,'' {\em The Annals
  of Mathematical Statistics}, vol.~22, no.~1, pp.~79--86, 1951.

\bibitem{Wehrl1978}
A.~Wehrl, ``General properties of entropy,'' {\em Reviews of Modern Physics},
  vol.~50, no.~2, pp.~221--260, 1978.

\bibitem{Helstrom1969}
C.~W. Helstrom, ``Quantum detection and estimation theory,'' {\em Journal of
  Statistical Physics}, vol.~1, no.~2, pp.~231--252, 1969.

\bibitem{Nielsen2000}
M.~A. Nielsen and I.~L. Chuang, {\em Quantum Computation and Quantum
  Information}.
\newblock Cambridge University Press, 2000.

\bibitem{gilchrist_distance_2005}
A.~Gilchrist, N.~K. Langford, and M.~A. Nielsen, ``Distance measures to compare
  real and ideal quantum processes,'' {\em Physical Review A}, vol.~71, no.~6,
  p.~062310, 2005.

\bibitem{rastegin_sine_2006}
A.~E. Rastegin, ``Sine distance for quantum states.'' {arXiv:quant-ph/0602112},
  2006.

\bibitem{tomamichel_duality_2010}
M.~Tomamichel, R.~Colbeck, and R.~Renner, ``Duality between smooth min- and
  max-entropies,'' {\em {IEEE} Transactions on Information Theory}, vol.~56,
  no.~9, pp.~4674--4681, 2010.

\bibitem{fuchs_cryptographic_1999}
C.~Fuchs and J.~van~de Graaf, ``Cryptographic distinguishability measures for
  quantum-mechanical states,'' {\em IEEE Transactions on Information Theory},
  vol.~45, no.~4, pp.~1216--1227, 1999.

\bibitem{watrous_semidefinite_2009}
J.~Watrous, ``Semidefinite programs for completely bounded norms,'' {\em Theory
  of Computing}, vol.~5, pp.~217--238, 2009.

\bibitem{boyd_convex_2004}
S.~Boyd and L.~Vandenberghe, {\em Convex Optimization}.
\newblock Cambridge University Press, 2004.

\bibitem{Grotschel1993}
M.~Gr\"{o}tschel, L.~Lov\'{a}sz, and A.~Schrijver, {\em Geometric algorithms
  and combinatorial optimization}.
\newblock Springer-Verlag, 1993.

\bibitem{matthews_finite_2012}
W.~Matthews and S.~Wehner, ``Finite blocklength converse bounds for quantum
  channels,'' {\em {arXiv:1210.4722}}, Oct. 2012.

\bibitem{Renyi1961}
A.~R\'enyi, ``On measures of entropy and information,'' in {\em Fourth Berkeley
  Symposium on Mathematical Statistics and Probability}, pp.~547--561, 1961.

\bibitem{datta_min_2009}
N.~Datta, ``Min- and max- relative entropies and a new entanglement monoton
  e,'' {\em {IEEE} Transactions on Information Theory}, vol.~55,
  pp.~2816--2826, May 2009.

\bibitem{Koenig2009IEEE_OpMeaning}
R.~K\"onig, R.~Renner, and C.~Schaffner, ``The operational meaning of min- and
  max-entropy,'' {\em IEEE Transactions on Information Theory}, vol.~55, no.~9,
  pp.~4337--4347, 2009.

\bibitem{TSSR11}
M.~Tomamichel, C.~Schaffner, A.~Smith, and R.~Renner, ``Leftover hashing
  against quantum side information,'' {\em IEEE Transactions on Information
  Theory}, vol.~57, no.~8, pp.~5524--5535, 2011.

\bibitem{EmailMarcoTom}
{M. Tomamichel, Personal communication}, 2013.

\bibitem{hiai_proper_1991}
F.~Hiai and D.~Petz, ``The proper formula for relative entropy and its
  asymptotics in quantum probability,'' {\em Communications in Mathematical
  Physics}, vol.~143, no.~1, pp.~99--114, 1991.

\bibitem{datta_strong_2011}
N.~Datta, M.~Mosonyi, M.-H. Hsieh, and F.~G. S.~L. Brandao, ``Strong converse
  capacities of quantum channels for classical information,'' {\em
  {arXiv:1106.3089}}, 2011.

\bibitem{berta_uncertainty_2010}
M.~Berta, M.~Christandl, R.~Colbeck, J.~M. Renes, and R.~Renner, ``The
  uncertainty principle in the presence of quantum memory,'' {\em Nature
  Physics}, vol.~6, pp.~659--662, 2010.

\bibitem{tomamichel_fully_2009}
M.~Tomamichel, R.~Colbeck, and R.~Renner, ``A fully quantum asymptotic
  equipartition property,'' {\em {IEEE} Transactions on Information Theory},
  vol.~55, no.~12, pp.~5840--5847, 2009.

\bibitem{tomamichel_framework_2012}
M.~Tomamichel, {\em A Framework for Non-Asymptotic Quantum Information Theory}.
\newblock {PhD}, {ETH} Zurich, 2012.
\newblock arXiv:1203.2142.

\end{thebibliography}

\end{document}